\documentclass[a4paper,onecolumn,10pt,accepted=2021-02-22]{quantumarticle}
\pdfoutput=1
\usepackage[utf8]{inputenc}
\usepackage[english]{babel}
\usepackage[T1]{fontenc}
\usepackage{amsmath}
\usepackage[numbers,sort&compress]{natbib}
\usepackage{tikz}
\usepackage{lipsum}

\usepackage{authblk}

\usepackage{xcolor}
\usepackage{amsmath,amsthm,amssymb,mathrsfs}
\usepackage{hyperref}
\usepackage{comment}
\usepackage{enumerate}
\usepackage{mathtools}
\usepackage{cases}

\usepackage{soul}

\usepackage{dutchcal}

\newcommand{\be}{\begin{equation}}
\newcommand{\ee}{\end{equation}}
\newcommand{\bea}{\begin{eqnarray}}
\newcommand{\eea}{\end{eqnarray}}
\newcommand{\bes}{\begin{equation*}}
\newcommand{\ees}{\end{equation*}}
\newcommand{\beas}{\begin{eqnarray*}}
\newcommand{\eeas}{\end{eqnarray*}}

\makeatletter
\newtheorem*{rep@theorem}{\rep@title}
\newcommand{\newreptheorem}[2]{%
\newenvironment{rep#1}[1]{%
 \def\rep@title{#2 \ref{##1} (restated)}%
 \begin{rep@theorem}}%
 {\end{rep@theorem}}}
\makeatother

\allowdisplaybreaks

\newtheorem{thm}{Theorem}
\newtheorem*{thm*}{Theorem}
\newtheorem{cor}[thm]{Corollary}

\newtheorem*{lem*}{Lemma}

\newtheorem{prop}[thm]{Proposition}
\newtheorem{defn}[thm]{Definition}

\newtheorem{fact}[thm]{Fact}

\newreptheorem{thm}{Theorem}
\newreptheorem{lem}{Lemma}
\newreptheorem{cor}{Corollary}
\theoremstyle{remark}
\newtheorem*{rem}{Remark}

\begin{document}

\title{Quantum Random Access Codes for Boolean Functions}

\author{Jo\~{a}o F. Doriguello}
\affiliation{Quantum Engineering Centre for Doctoral Training, University of Bristol, United Kingdom}
\affiliation{School of Mathematics, University of Bristol, United Kingdom}
\email{joaof.doriguello@gmail.com}
\homepage{http://joaodoriguello.com}

\author{Ashley Montanaro}
\affiliation{School of Mathematics, University of Bristol, United Kingdom}
\affiliation{Phasecraft Ltd.}

	\maketitle

	\begin{abstract}
	    An $n\overset{p}{\mapsto}m$ random access code (RAC) is an encoding of $n$ bits into $m$ bits such that any initial bit can be recovered with probability at least $p$, while in a quantum RAC (QRAC), the $n$ bits are encoded into $m$ qubits. Since its proposal, the idea of RACs was generalized in many different ways, e.g.\ allowing the use of shared entanglement (called entanglement-assisted random access code, or simply EARAC) or recovering multiple bits instead of one. In this paper we generalize the idea of RACs to recovering the value of a given Boolean function $f$ on any subset of fixed size of the initial bits, which we call $f$-random access codes. We study and give protocols for $f$-random access codes with classical ($f$-RAC) and quantum ($f$-QRAC) encoding, together with many different resources, e.g.\ private or shared randomness, shared entanglement ($f$-EARAC) and Popescu-Rohrlich boxes ($f$-PRRAC). The success probability of our protocols is characterized by the \emph{noise stability} of the Boolean function $f$. Moreover, we give an \emph{upper bound} on the success probability of any $f$-QRAC with shared randomness that matches its success probability up to a multiplicative constant (and $f$-RACs by extension), meaning that quantum protocols can only achieve a limited advantage over their classical counterparts.
	\end{abstract}

\section{Introduction}

One of the possible origins of quantum computers' power is the exponential size of the Hilbert space: a $n$-qubit quantum state is a unit vector in a $2^n$ dimensional complex vector space. On the other hand, one of the fundamental results in quantum information theory -- Holevo's theorem~\cite{holevo1973bounds} -- states that no more than $n$ bits of classical information can be transmitted by $n$ qubits without entanglement. Nonetheless, interesting scenarios arise when allowing a small chance of transmitting the wrong message or/and obtaining \emph{partial} information at the expense of losing information about the rest of the system. One of these scenarios is the concept of \emph{quantum random access codes} (QRACs), where a number of bits are encoded into a smaller number of qubits such that any one of the initial bits can be recovered with some probability of success. A QRAC is normally denoted by $n\overset{p}{\mapsto} m$, meaning that $n$ bits are encoded into $m$ qubits such that any initial bit can be recovered with probability at least $p > 1/2$ (greater than $1/2$ since $p = 1/2$ can be achieved by pure guessing), and a classical version, called simply \emph{random access code} (RAC), is similarly defined, with the encoding message being $m$ bits. The idea of QRACs first appeared in a paper by Stephen Wiesner~\cite{wiesner1983conjugate} in 1983 under the name of \emph{conjugate coding}, and was later rediscovered by Ambainis \emph{et al.}\ in 1999~\cite{ambainis1999dense}.

Quantum random access codes found application in many different contexts, e.g.\ quantum finite automata~\cite{ambainis1999dense,nayak1999optimal}, network coding~\cite{hayashi20064,hayashi2007quantum}, quantum communication complexity~\cite{buhrman2001communication,MR2506518,MR2115303}, locally decodable codes~\cite{ben2008hypercontractive,kerenidis2004quantum,MR2087942,MR2184730}, non-local games~\cite{muhammad2014quantum,tavakoli2016spatial}, cryptography~\cite{pawlowski2011semi}, quantum state learning~\cite{MR2386653}, device-independent dimension witnessing~\cite{aguilar2018certifying,ahrens2014experimental,wehner2008lower}, self-testing measurements~\cite{farkas2020self,farkas2019self}, randomness expansion~\cite{li2011semi}, studies of no-signaling resources~\cite{grudka2014popescu}, and characterization of quantum correlations from information theory~\cite{pawlowski2009information}. The $2\mapsto 1$ and $3\mapsto 1$ QRACs were first experimentally demonstrated in~\cite{spekkens2009preparation}. See~\cite{foletto2020experimental,hameedi2017complementarity,muhammad2014quantum,tavakoli2015quantum,wang2019experimental} for subsequent demonstrations.

In this paper we further generalize the idea of (quantum) random access codes to recovering not just an initial bit, but the value of a fixed Boolean function on any subset of the initial bits with fixed size. We call them $f$-random access codes. The case of the Parity function was already considered in~\cite{ben2008hypercontractive}, and here we generalize to arbitrary Boolean functions $f:\{-1,1\}^k\to\{-1,1\}$.

\subsection{Related Work}

An $n\overset{p}{\mapsto}m$ (Q)RAC is an encoding of $n$ bits into $m$ (qu)bits such that any initial bit can be recovered with probability at least $p$. This probability is the \emph{worst case success probability} over all possible pairs $(x,i)$ of input string $x\in\{-1,1\}^n$ and recoverable bit $i\in\{1,\dots,n\}$. Many different resources can be used during the encoding and decoding, e.g.\ private randomness~(PR), shared randomness~(SR), shared entanglement, and even super-quantum correlations like Popescu-Rohrlich boxes~\cite{popescu1994quantum}. 

Regarding the classical RAC, Ambainis \emph{et al.}~\cite{ambainis1999dense} proved that there is no $2\overset{p}{\mapsto}1$ RAC (and $2^m\overset{p}{\mapsto}m$ RAC by extension) with PR and worst case success probability $p>1/2$. On the other hand, Ambainis \emph{et al.}~\cite{ambainis2008quantum} showed that RACs with SR can achieve success probability $p>1/2$.
\begin{thm}[{\cite[Equation~(25)]{ambainis2008quantum}}]
    \label{thr:thr2.1}
    The optimal $n\overset{p}{\mapsto}1$ RAC with SR has success probability
    \begin{align*}
        p = \frac{1}{2} + \frac{1}{2^n}\binom{n-1}{\lfloor \frac{n-1}{2}\rfloor} = \frac{1}{2} + \frac{1}{\sqrt{2\pi n}} - O\left(\frac{1}{n^{3/2}}\right).
    \end{align*}
\end{thm}
For a general number of encoded bits, Ambainis \emph{et al.}~\cite{ambainis1999dense} developed a RAC with PR using a specific code from~\cite{cohen1983nonconstructive} which matches their classical lower bound $m \geq (1-H(p))n$ up to an additive logarithmic term, where $H(p) = -p\log_2{p} - (1-p)\log_2(1-p)$ is the binary entropy function.
\begin{thm}[{\cite[Theorem 2.2]{ambainis1999dense}}]
    \label{thr:thr2.1a}
    There is an $n\overset{p}{\mapsto} m$ RAC with PR and $m \leq (1-H(p))n + 7\log_2{n}$ for any $p > \frac{1}{2}$.
\end{thm}

As for QRACs, Ambainis \emph{et al.}~\cite{ambainis1999dense} showed the existence of a $2\overset{p}{\mapsto}1$ QRAC with PR\footnote{Usually private randomness is already assumed in QRACs under the encoding onto density matrices.} and $p = \frac{1}{2} + \frac{1}{2\sqrt{2}}\approx 0.85$, and the existence of a $3\overset{p}{\mapsto}1$ QRAC with PR and $p = \frac{1}{2} + \frac{1}{2\sqrt{3}} \approx 0.79$ (the second attributed to Chuang). Later Hayashi \emph{et al.}~\cite{hayashi20064} showed the impossibility of a $4\overset{p}{\mapsto} 1$ QRAC (and $4^m\overset{p}{\mapsto} m$ QRAC by extension) with PR and success probability $p>1/2$. Similarly to the classical case, Ambainis \emph{et al.}~\cite{ambainis2008quantum} showed that QRACs can also benefit from SR.
\begin{thm}[{\cite[Theorem 6]{ambainis2008quantum}}]
    \label{thr:thr2.2}
    There is an $n\overset{p}{\mapsto}1$ QRAC with SR and \footnote{Ambainis \emph{et al.}~\cite{ambainis2008quantum} do not give the high order terms, but these can be calculated by following their procedure together with~\cite[Equation 2.198]{hughes1995random}.}
    \begin{align*}
        p = \frac{1}{2} + \sqrt{\frac{2}{3\pi n}} + O\left(\frac{1}{n^{3/2}}\right).
    \end{align*}
\end{thm}
The specific case of $m=2$ encoding qubits was explored in~\cite{hayashi20064,imamichi2018constructions,liabotro2017improved}. For the general case of $m>1$, Iwama \emph{et al.}~\cite{iwama2007unbounded} constructed an $(4^m-1)\overset{p}{\mapsto}m$ QRAC with PR and $p = \frac{1}{2} + \frac{1}{2(2^m-1)\sqrt{2^m+1}}$ (such construction also works for all $n<4^m$). On the other hand, Ambainis \emph{et al.}~\cite{ambainis1999dense} proved that if an $n\overset{p}{\mapsto} m$ QRAC with PR and $p>1/2$ exists, then $m = \Omega(n/\log{n})$, which was later improved to $m\geq(1-H(p))n$ by Nayak~\cite{nayak1999optimal}, thus matching the same classical lower bound from~\cite{ambainis1999dense}.

The idea of decoding a function of the initial bits instead of a single bit was already considered by Ben-Aroya, Regev and de Wolf~\cite{ben2008hypercontractive} (who also considered recovering multiple bits rather than just one). More specifically, they defined an $n\overset{p}{\mapsto}m$ $\operatorname{XOR}_k$-QRAC, where $n$ bits are encoded into $m$ qubits such that the parity of any $k$ initial bits can be recovered with success probability at least~$p$.\footnote{In their definition the success probability is the average over random $k$-subsets and random inputs, which, in our context, is equivalent to using SR.} Using their hypercontractive inequality for matrix-valued functions, they proved the following upper bound on the success probability.
\begin{thm}[{\cite[Theorem 7]{ben2008hypercontractive}}]
    \label{thr:thr1.4}
    For any $\eta > 2\ln{2}$ there is a constant $C_\eta$ such that, for any $n\overset{p}{\mapsto}m$ $\operatorname{XOR}_k$-QRAC with SR and $k=o(n)$,
    \begin{align}
        p \leq \frac{1}{2} + C_\eta \left(\frac{\eta m}{n}\right)^{k/2}.\label{eq:eq1.4}
    \end{align}
\end{thm}
They conjectured that the factor $\eta > 2\ln{2}$ can be dropped from the above bound, and thus extended to $m/n > 1/(2\ln{2}) \approx 0.72$, although it might require a strengthening of their hypercontractive inequality.

The use of shared entanglement in random access codes was first considered by Klauck~\cite{MR2115303,klauck2007one}. Here the encoding and decoding parties are allowed to use an arbitrary amount of shared entangled states (note that shared entanglement can be used to obtain both private and shared randomness). The figure of merit in this generalization is the relation between $n$, $m$ and $p$, while the amount of shared entanglement is not taken into account. Klauck~\cite{MR2115303,klauck2007one} considered an $n\overset{p}{\mapsto}m$ QRAC with shared entanglement and, by its equivalence to the quantum one-way communication complexity for the index function, proved the lower bound $m\geq (1-H(p))n/2$, similar to Nayak's bound. Later Paw\l{}owski and \.{Z}ukowski~\cite{pawlowski2010entanglement} coined the term entanglement-assisted random access code (EARAC), which is a RAC with shared entanglement, and studied the case when $m=1$, giving protocols with better decoding probabilities compared to the usual $n\overset{p}{\mapsto} 1$ QRAC with SR. Recently T{\u{a}}n{\u{a}}sescu \emph{et al.}~\cite{tuanuasescu2020optimal} expanded the idea of $n\overset{p}{\mapsto}1$ EARACs to recovering an initial bit under a specific request distribution.

\begin{thm}[\cite{pawlowski2010entanglement} and {\cite[Corollary 2 and Theorem 5]{tuanuasescu2020optimal}}]
    \label{thr:thr2.3}
    The optimal $n\overset{p}{\mapsto}1$ EARAC with SR has success probability
    \begin{align*}
        p = \frac{1}{2} + \frac{1}{2\sqrt{n}}.
    \end{align*}
\end{thm}

The idea of (Q)RAC was generalized in other ways, e.g.\ parity-oblivious~\cite{ambainis2019parity,chailloux2016optimal,spekkens2009preparation} and multiparty~\cite{saha2020multiparty} versions, encoding on $d$-valued qubits (qudits)~\cite{ambainis2015optimal,casaccino2008extrema,farkas2019self,liabotro2017improved,tavakoli2015quantum}, a wider range of information retrieval tasks~\cite{emeriau2020quantum} and a connection to Popescu-Rohrlich boxes. It was shown~\cite{wolf2005oblivious} that a Popescu-Rohrlich box can simulate a RAC by means of just one bit of communication, while in~\cite{grudka2014popescu} the converse was proven. An object called \emph{racbox}~\cite{grudka2014popescu} was defined, which is a box that implements a RAC when supported with one bit of communication, and it was shown that a \emph{non-signaling} racbox is equivalent to a Popescu-Rohrlich box. A quantum version of a racbox was later proposed in~\cite{grudka2015nonsignaling}. Finally, we mention that RACs were also studied within ``theories'' that violate the uncertainty relation for anti-commuting observables and present stronger-than-quantum correlations~\cite{steeg2008relaxed}.

\subsection{Our Results}

This paper focuses on generalizing the classical, quantum and entanglement-assisted random access codes. Instead of recovering a single bit from the initial string $x\in\{-1,1\}^n$, we are interested in evaluating a Boolean function $f:\{-1,1\}^k\to\{-1,1\}$ on any sequence of $k$ bits from $x$. We generically call them $f$-random access codes. Let $\mathcal{S}_n^k = \{(S_i)_{i=1}^k\in \{1,\dots,n\}^k ~|~ S_i \neq S_j ~\forall i,j\}$ be the set of sequences of different elements from $\{1,\dots,n\}$ with length $k$ and let $x_S\in\{-1,1\}^k$ denote the substring of $x\in\{-1,1\}^n$ specified by $S\in \mathcal{S}_n^k$. Alice gets $x\in\{-1,1\}^n$ and she needs to encode her data and send it to Bob, so that he can decode $f(x_S)$ for any $S\in\mathcal{S}_n^k$ with probability $p>1/2$. Such problem was already considered by Sherstov in a \emph{two-way communication complexity} setting~\cite{sherstov2009separating} and later used in his pattern matrix method~\cite{sherstov2011pattern} in order to prove other communication complexity lower bounds. Even though our results are expressed in a random access code language, they can also be seen as in a \emph{one-way communication complexity} setting. If two-way communication is allowed, Bob can send the identity of his sequence to Alice with $O(k \log n)$ bits of communication, whereas (as we will see) significantly more communication may be required in the one-way scenario.

In the following, $\Pi$ will refer to a sample space with some probability distribution. As before, PR and SR stand for private and shared randomness, respectively. Moreover, since we require the success probability to always be greater than $1/2$, given that one can always guess the correct result with probability $1/2$, from now on it will be convenient to use the \emph{bias} $\varepsilon$ of the prediction, defined as $\varepsilon = 2p - 1$, instead of its success probability $p$.

We start with $n\overset{\varepsilon}{\mapsto}m$ $f$-RAC, the $f$-classical random access code on $m$ bits with bias $\varepsilon$.
\begin{defn}
    An $n\overset{\varepsilon}{\mapsto}m$ $f$-RAC with PR is an encoding map $E:\{-1,1\}^n\times\Pi_A\to\{-1,1\}^m$ satisfying the following: for every $S\in\mathcal{S}_n^k$ there is a decoding map $D_{S}:\{-1,1\}^m\times\Pi_B \to\{-1,1\}$ such that $\operatorname{Pr}_{r_A,r_B}[D_{S}(E(x,r_A),r_B) = f(x_S)] \geq \frac{1}{2} + \frac{1}{2}\varepsilon$ for all $x\in\{-1,1\}^n$.
\end{defn}
\begin{defn}
    An $n\overset{\varepsilon}{\mapsto}m$ $f$-RAC with SR is an encoding map $E:\{-1,1\}^n\times\Pi\to\{-1,1\}^m$ satisfying the following: for every $S\in\mathcal{S}_n^k$ there is a decoding map $D_{S}:\{-1,1\}^m\times\Pi\to\{-1,1\}$ such that $\operatorname{Pr}_{r}[D_{S}(E(x,r),r) = f(x_S)] \geq \frac{1}{2} + \frac{1}{2}\varepsilon$ for all $x\in\{-1,1\}^n$.
\end{defn}
We define $n\overset{\varepsilon}{\mapsto}m$ $f$-QRAC, the $f$-quantum random access code on $m$ qubits with bias $\varepsilon$.
\begin{defn}
    An $n\overset{\varepsilon}{\mapsto}m$ $f$-QRAC with PR is an encoding map $E:\{-1,1\}^n\to\mathbb{C}^{2^m\times 2^m}$ that assigns an $m$-qubit density matrix to every $x\in\{-1,1\}^n$ and satisfies the following: for every $S\in\mathcal{S}_n^k$ there is a POVM $M^S = \{M^S_{-1}, M^S_{1}\}$ such that $\operatorname{Tr}(M^S_{f(x_S)}\cdot E(x)) \geq \frac{1}{2} + \frac{1}{2}\varepsilon$ for all $x\in\{-1,1\}^n$.
\end{defn}
\begin{defn}
    \label{def:def9}
    An $n\overset{\varepsilon}{\mapsto}m$ $f$-QRAC with SR is an encoding map $E:\{-1,1\}^n\times \Pi\to\mathbb{C}^{2^m}$ that assigns an $m$-qubit pure state to every $x\in\{-1,1\}^n$ and satisfies the following: for every $S\in\mathcal{S}_n^k$ there is a set of POVMs $\{M^S_r\}_{r\in\Pi}$, with $M^S_{r} = \{M^S_{-1,r}, M^S_{1,r}\}$, such that $\mathbb{E}_{r}[E(x,r)^\dagger M^S_{f(x_S),r} E(x,r)] \geq \frac{1}{2} + \frac{1}{2}\varepsilon$ for all $x\in\{-1,1\}^n$.
\end{defn}
Similarly, we define $n\overset{\varepsilon}{\mapsto}m$ $f$-EARAC, the $f$-entanglement-assisted random access code on $m$ bits with bias $\varepsilon$.
\begin{defn}
    \label{def:def10}
    An $n\overset{\varepsilon}{\mapsto}m$ $f$-EARAC is an $n\overset{\varepsilon}{\mapsto}m$ $f$-RAC with SR where the encoding and decoding parties share an unlimited amount of entangled quantum states.
\end{defn}
\noindent Due to shared entanglement being a source of SR, we already include SR in $f$-EARACs. We note that~\cite{pawlowski2010entanglement} focused on EARACs without SR.

Finally, we define $n\overset{\varepsilon}{\mapsto}m$ $f$-PRRAC, the $f$-Popescu-Rohrlich random access code on $m$ bits with bias $\varepsilon$. A Popescu-Rohrlich box~\cite{popescu1994quantum} is a bipartite system shared by two parties with two inputs $x,y\in\{0,1\}$ and two outputs $a,b\in\{0,1\}$ and is defined by the joint probability distribution
\begin{align*}
    \operatorname{Pr}[ab|xy] = \begin{cases}
        \frac{1}{2} &~\text{for}~ a\oplus b=xy,\\
        0 &~\text{otherwise}.
    \end{cases}
\end{align*}
\begin{defn}
    An $n\overset{\varepsilon}{\mapsto}m$ $f$-PRRAC is an $n\overset{\varepsilon}{\mapsto}m$ $f$-RAC with SR where the encoding and decoding parties share an unlimited amount of Popescu-Rohrlich boxes.
\end{defn}

In Section~\ref{sec:sec4} we devise encoding-decoding strategies for all the $f$-random access codes just defined, thus deriving \emph{lower bounds} on their biases given the encoding/decoding parameters $n,m$ and $k$. These $f$-random access codes are built based on previous ideas from~\cite{ambainis2008quantum,ambainis1999dense,pawlowski2010entanglement}. The Boolean function $f:\{-1,1\}^k\to\{-1,1\}$ that needs to be evaluated directly influences the final bias and such influence in our results is captured by the single quantity called \emph{noise stability}~\cite{benjamini1999noise,o2003computational}. Informally it is a measure of how resilient to noise a Boolean function is. Given a uniformly random input $x\in\{-1,1\}^n$, one might imagine a process that flips each bit of $x$ independently with some probability $\frac{1}{2} - \frac{1}{2}q$, where $q\in[-1,1]$, which leads to some final string $y\in\{-1,1\}^n$. The noise stability $\operatorname{Stab}_q[f]$ of $f$ with parameter $q$ is the correlation between $f(x)$ and $f(y)$ (see Section~\ref{sec:sec3} for a formal definition).

Our positive results can be summarized by the following theorem.
\begin{thm}
    \label{thr:thr1.3}
    Let $f:\{-1,1\}^k\to\{-1,1\}$ be a Boolean function and $\operatorname{Stab}_q[f]$ its noise stability with parameter $q$.
    \begin{enumerate}[(a)]
        \item Let $\ell\in\mathbb{N}$. If $m=\Omega(\ell\log{n})$ and $k=o(\sqrt{\ell})$, there is an $n\overset{\varepsilon}{\mapsto}m$ $f$-RAC with PR and bias $\varepsilon \geq (1-o_n(1))\operatorname{Stab}_q[f]$ with $q \geq \sqrt{\frac{m}{n} - \frac{5\log_2{(n/\ell)}}{n/\ell}}$.
        \item If $k=o(\sqrt{m})$, there is an $n\overset{\varepsilon}{\mapsto}m$ $f$-RAC with SR and $\varepsilon \geq (1-o_n(1)) \operatorname{Stab}_q[f]$ with $q \geq \sqrt{\frac{m}{2 n}}$.
        \item If $k=o(\sqrt{m})$, there is an $n\overset{\varepsilon}{\mapsto}m$ $f$-QRAC with SR and $\varepsilon \geq (1-o_n(1)) \operatorname{Stab}_q[f]$ with $q \geq \sqrt{\frac{8m}{3\pi n}}$.
        \item If $k=o(\sqrt{m})$, there is an $n\overset{\varepsilon}{\mapsto}m$ $f$-EARAC with $\varepsilon \geq (1-o_n(1)) \operatorname{Stab}_q[f]$ and $q = \sqrt{\frac{m}{n}}$.
        \item For any $n\in\mathbb{N}$, there is an $n\overset{1}{\mapsto}1$ $f$-PRRAC.
    \end{enumerate}
\end{thm}
Results (a), (b), (c) and (d) use an encoding scheme reminiscent of the concatenation idea from \cite{pawlowski2009information,pawlowski2010entanglement,tuanuasescu2020optimal} (and suggested to us by Ronald de Wolf). The underlying idea is to randomly break the initial string $x\in\{-1,1\}^n$ into different `blocks' and encode them via a standard RAC/QRAC/EARAC. Result (a) breaks $x$ into $\ell$ blocks and employs the $n/\ell\mapsto m/\ell$ RAC from Theorem~\ref{thr:thr2.1a} on every block, each with $n/\ell$ elements, while in results (b)/(c)/(d) we employ the $n/m\mapsto 1$ RAC/QRAC/EARAC from Theorems~\ref{thr:thr2.1}/\ref{thr:thr2.2}/\ref{thr:thr2.3} in order to encode $m$ blocks, each with $n/m$ elements, into a single (qu)bit each, resulting in $m$ encoded (qu)bits. With high probability all the bits from the needed string $x_S\in\{-1,1\}^k$ will be encoded into different blocks and therefore can be decoded and $f$ evaluated. The decoded string $y\in\{-1,1\}^k$ can be viewed as a `noisy' $x_S$, to which the noise stability framework can be applied. The bias of the base RAC/QRAC/EARAC thus becomes the parameter $q$ in the noise stability of the corresponding $f$-random access code. As a quick remark, since we opted to lower-bound the parameters $q$ in Theorem~\ref{thr:thr1.3}, in result~(b) $q$ does not exactly equal the bias from Theorem~\ref{thr:thr2.1}. One could write, though, $q\approx\sqrt{\frac{2m}{\pi n}}$. 

Result (a) is our strongest bound, since it also applies to all other $f$-random access codes. Moreover, there is some freedom in setting the number of blocks $\ell$, since the number of encoded bits in Theorem~\ref{thr:thr2.1a} is not fixed to a single number (as opposed to Theorems~\ref{thr:thr2.1},~\ref{thr:thr2.2} and~\ref{thr:thr2.3}). The result is a trade-off between the number of bits $k$ of the Boolean function and the number of encoded bits $m$. However, the number of encoded bits in result (a) is limited to $m=\Omega(\log{n})$, a characteristic inherited from the RAC in Theorem~\ref{thr:thr2.1a}. It is possible to go below this limit by using SR, as demonstrated by results (b), (c) and (d).

The above results show that quantum resources offer a modest advantage over the classical $f$-random access code. On the other hand, result (e) demonstrates that stronger-than-quantum resources like Popescu-Rohrlich boxes can lead to extremely powerful $f$-random access codes. This is a consequence of violating Information Causality~\cite{pawlowski2009information}, since one bit transfer allows the access to \emph{any} bit in a database via Popescu-Rohrlich boxes. From $x\in\{-1,1\}^n$ a long bit-string $x_f\in\{-1,1\}^t$, where $t = |\mathcal{S}_n^k|$, can be constructed with the values $f(x_S)$ for all $S\in\mathcal{S}_n^k$. All bits from $x_f$ are readable with the aid of Popescu-Rohrlich boxes, with non-signaling constraining the readout to just one bit. The protocol for $f$-PRRACs is taken from~\cite{pawlowski2009information} and uses a pyramid of Popescu-Rohrlich boxes and nests a van Dam's protocol~\cite{van2013implausible}.

In Section~\ref{sec:sec5} we prove an \emph{upper bound} on the bias of any $f$-QRAC with SR (and $f$-RAC) using the same method of the hypercontractive inequality for matrix-valued functions from~\cite{ben2008hypercontractive}.
\begin{thm}
    \label{thr:thr1.5}
    Let $f:\{-1,1\}^k\to\{-1,1\}$ be a Boolean function. For any $n\overset{\varepsilon}{\mapsto}m$ $f$-QRAC with SR and $k = o(n)$ the following holds: for any $\eta > 2\ln{2}$ there is a constant $C_\eta$ such that
    \begin{align}
        \varepsilon \leq C_\eta \sum_{\ell = 0}^k L_{1,\ell}(f)\left(\frac{\eta m}{n}\right)^{\ell/2}, \label{eq:eq1.5}
    \end{align}
    where $L_{1,\ell}(f) = \sum_{\substack{T\subseteq[k]\\ |T| = \ell}}\big|\widehat{f}(T)\big|$ is the $1$-norm of the $\ell$-th level of the Fourier transform of $f$.
\end{thm}
One can see that the above result is a generalization of Theorem~\ref{thr:thr1.4}. Indeed, for Parity on $k$ bits, $L_{1,\ell}(\operatorname{XOR}_k) = 1$ iff $\ell=k$, and so Eq.~(\ref{eq:eq1.4}) is recovered. The following corollary from Theorem~\ref{thr:thr1.5} helps to compare the bias upper bound to the bias lower bounds from Theorem~\ref{thr:thr1.3}.
\begin{cor}
    Let $f:\{-1,1\}^k\to\{-1,1\}$ be a Boolean function. For any $n\overset{\varepsilon}{\mapsto}m$ $f$-QRAC with SR and $k = o(n)$ the following holds: for any $\eta > 2\ln{2}$ there is a constant $C_\eta$ such that
    \begin{align*}
         \varepsilon \leq C_\eta 2^{\operatorname{deg}(f)-1}\operatorname{Stab}_q[f],
    \end{align*}
    where $\operatorname{deg}(f) = \max\{|S|:\widehat{f}(S)\neq 0\}$ is the degree of $f$ and $q=\sqrt{\frac{\eta m}{n}}$.
\end{cor}
Taking $\operatorname{deg}(f)$ to be upper-bounded by a constant (for example, if $k=O(1)$), our bias upper bound matches our bias lower bounds for $f$-RAC/QRAC with SR up to a global multiplicative constant and a multiplicative constant $\sqrt{\eta}$ in the parameter~$q$. We conjecture that the parameter $q$ can be improved to $\sqrt{\frac{m}{n}}$, which might require a stronger version of the hypercontractive inequality.
Other corollaries from Theorem~\ref{thr:thr1.5} are derived in Section~\ref{sec:sec5} and compared to our bias lower bounds.

Upper bound~(\ref{eq:eq1.5}) does not apply to $f$-EARACs. Previously, it was known that for the special case of standard EARACs ($m=1$), the bias $\varepsilon$ is upper-bounded by $1/\sqrt{n}$ (Theorem~\ref{thr:thr2.3}). This upper bound can generalised to EARACs with $m>1$ assuming an independence condition (Section~\ref{sec:sec5.3}). The resulting bound is $\varepsilon \le \sqrt{m/n}$. We view this as evidence that the bias lower bound for the general case of $f$-EARACs given in Theorem~\ref{thr:thr1.3} should actually be tight.

Regarding the quantity $\operatorname{Stab}_q[f]$ itself, it can be nicely related to the Fourier coefficients of $f$ (see Theorem~\ref{thr:thr2.0} in the next section). We briefly mention the noise stability for a few functions. For Parity ($\operatorname{XOR}_k$), $\operatorname{Stab}_q[\operatorname{XOR}_k] = q^k$, and, more generally, for any function $\chi_S(x) = \prod_{i\in S}x_i$, $\operatorname{Stab}_q[\chi_S] = q^{|S|}$. As for the Majority function ($\operatorname{MAJ}_k$), one can show that~\cite[Theorem~5.18]{o2014analysis} $\frac{2}{\pi}\arcsin{q} \leq \operatorname{Stab}_q[\operatorname{MAJ}_k] \leq \frac{2}{\pi}\arcsin{q} + O\left(\frac{1}{\sqrt{1-q^2}\sqrt{k}}\right)$. Other examples can be found in~\cite{mossel2003noise}. Moreover, a randomized algorithm for approximating the noise stability of monotone Boolean functions up to relative error was proposed in~\cite{rubinfeld2019approximating}.

In Section~\ref{sec:sec4} we present protocols for all our $f$-random access codes, and in Section~\ref{sec:sec5} we derive an upper bound on the bias of $f$-QRACs with SR.

\section{Preliminaries}
\label{sec:sec3}

We shall briefly revise some results from Boolean analysis that are going to be useful. For an introduction to the analysis of Boolean functions, see O'Donnell's book~\cite{o2014analysis} or de Wolf's paper~\cite{de2008brief}. In the following, we write $[n] = \{1,\dots,n\}$ and $\mathbb{S}_n$ is the set of all permutations of $[n]$. As before, let $\mathcal{S}_n^k = \{(S_i)_{i=1}^k\in [n]^k ~|~ S_i \neq S_j ~\forall i,j\}$ be the set of sequences of different elements from $[n]$ with length $k$ and let $x_S\in\{-1,1\}^k$ denote the substring of $x\in\{-1,1\}^n$ specified by $S\in \mathcal{S}_n^k$.

The inner product $\langle \cdot,\cdot\rangle$ on the vector space of all functions $f:\{-1,1\}^n\to\mathbb{R}$ is defined by
\begin{align*}
    \langle f,g\rangle = \frac{1}{2^n}\sum_{x\in\{-1,1\}^n}f(x)g(x).
\end{align*}
Every function $f:\{-1,1\}^n\to\mathbb{R}$ can be uniquely expressed as a multilinear polynomial, its \emph{Fourier expansion}, as
\begin{align*}
    f(x) = \sum_{S\subseteq[n]}\widehat{f}(S)\chi_S(x),
\end{align*}
where, for $S\subseteq[n]$, $\chi_S:\{-1,1\}^n\to\{-1,1\}$ is defined by
\begin{align*}
    \chi_S(x) = \prod_{i\in S}x_i.
\end{align*}
The real number $\widehat{f}(S)$ is called the Fourier coefficient of $f$ on $S$ and is given by
\begin{align*}
    \widehat{f}(S) = \langle f,\chi_S\rangle = \frac{1}{2^n}\sum_{x\in\{-1,1\}^n}f(x)\chi_S(x).
\end{align*}

An important and useful concept for Boolean functions is \emph{noise stability}. As previously mentioned, it is a measure of how resilient to noise a Boolean function is, and is defined from the concept of $q$-correlated pairs of random strings given below.
\begin{defn}[{\cite[Definitions 2.40 and 2.41]{o2014analysis}}]
    \label{def:def1}
    Let $q\in[-1,1]$. For fixed $x\in\{-1,1\}^n$ we write $y\sim N_q(x)$ to denote that the random string $y$ is drawn as follows: for each $i\in[n]$ independently,
    \begin{align*}
        y_i = \begin{cases}
            x_i &\text{with probability } \frac{1}{2} + \frac{1}{2}q,\\
            -x_i &\text{with probability } \frac{1}{2} - \frac{1}{2}q.
        \end{cases}
    \end{align*}
    We say that $y$ is $q$-correlated to $x$. If $x\sim\{-1,1\}^n$ is drawn uniformly at random and then $y\sim N_q(x)$, we say that $(x,y)$ is a $q$-correlated pair of random strings.
\end{defn}
Given these definitions, we can formally define the concept of noise stability, which measures the correlation between $f(x)$ and $f(y)$ when $(x,y)$ is a $q$-correlated pair.
\begin{defn}[{\cite[Definition 2.42]{o2014analysis}}]
    For $f:\{-1,1\}^n\to\mathbb{R}$ and $q\in[-1,1]$, the \emph{noise stability} of $f$ at $q$ is
    \begin{align*}
        \operatorname{Stab}_q[f] = \operatorname*{\mathbb{E}}_{\substack{(x,y) \\ q\operatorname{-correlated}}}\left[f(x)f(y)\right].
    \end{align*}
\end{defn}
The noise stability of $f$ is nicely related to $f$'s Fourier coefficients as stated in the following theorem.
\begin{thm}[{\cite[Theorem 2.49]{o2014analysis}}]
    \label{thr:thr2.0}
    For $f:\{-1,1\}^n\to\mathbb{R}$ and $q\in[-1,1]$,
    \begin{align*}
        \operatorname{Stab}_q[f] = \sum_{k=0}^n q^k W^k[f],
    \end{align*}
    where $W^k[f] = \sum_{\substack{S\subseteq[n]\\|S|=k}}\widehat{f}(S)^2$ is the Fourier weight of $f$ at degree $k$.
\end{thm}
\noindent The above result makes it clear that $\operatorname{Stab}_q[f]$ is an increasing function of $q$ for $q \ge 0$.

Theorem~\ref{thr:thr2.0} is obtained from one of the most important operators in analysis of Boolean functions: the \emph{noise operator} $\operatorname{T}_q$.
\begin{defn}[{\cite[Definition 2.46]{o2014analysis}}]
    For $q\in[-1,1]$, the \emph{noise operator with parameter $q$} is the linear operator $\operatorname{T}_q$ on functions $f:\{-1,1\}^n\to\mathbb{R}$ defined by
    \begin{align*}
        \operatorname{T}_q f(x) = \operatorname*{\mathbb{E}}_{y\sim N_q(x)}[f(y)].
    \end{align*}
\end{defn}
\begin{prop}[{\cite[Proposition 2.47]{o2014analysis}}]
    \label{prop:prop1}
    For $f:\{-1,1\}^n\to\mathbb{R}$, the Fourier expansion of $\operatorname{T}_q f$ is
    \begin{align*}
        \operatorname{T}_q f = \sum_{S\subseteq[n]}q^{|S|}\widehat{f}(S)\chi_S.
    \end{align*}
\end{prop}
\noindent It is not hard to prove from the above results that $\operatorname{Stab}_q[f] = \langle f,\operatorname{T}_q f\rangle$.

Some of the above concepts can be generalized to matrix-valued functions. The \emph{Fourier transform} $\widehat{f}$ of a matrix-valued function $f:\{-1,1\}^n\to\mathbb{C}^{m\times m}$ is defined similarly as for scalar functions: {\color{red}it} is the function $\widehat{f}:\{-1,1\}^n\to\mathbb{C}^{m\times m}$ defined by
\begin{align*}
     \widehat{f}(S) = \frac{1}{2^n}\sum_{x\in\{-1,1\}^n} f(x)\chi_S(x).
\end{align*}
Here the Fourier coefficients $\widehat{f}(S)$ are also $m\times m$ complex matrices. Moreover, given $A\in\mathbb{C}^{m\times m}$ with singular values $\sigma_1,\dots,\sigma_m$, its \emph{trace norm} is defined as $\|A\|_{\operatorname{tr}} = \operatorname{Tr}|A| = \sum_{i=1}^m \sigma_i$.

We shall make use of the following result from Ben-Aroya, Regev and de Wolf~\cite{ben2008hypercontractive}, which stems from their hypercontractive inequality for matrix-valued functions.
\begin{thm}[{\cite[Lemma 6]{ben2008hypercontractive}}]
    \label{thm:thm2.2.c1}
    For every $f:\{-1,1\}^n\to\mathbb{C}^{2^m\times 2^m}$ and $\delta\in[0,1]$,
    \begin{align*}
        \sum_{S\subseteq[n]} \delta^{|S|}\|\widehat{f}(S)\|_{\operatorname{tr}}^2 \leq 2^{2\delta m}.
    \end{align*}
\end{thm}

Finally, $H:[0,1]\to[0,1]$ given by $H(p) = -p\log_2{p} - (1-p)\log_2{(1-p)}$ is the binary entropy function. The following bounds hold.
\begin{thm}[{\cite[Theorem 2.2]{calabro2009exponential}}]
    \label{thr:thr3.1}
    $\forall p\in[0,1]$, $1 - 4\left(p - \frac{1}{2}\right)^2 \leq H(p) \leq 1 - \frac{2}{\ln{2}}\left(p - \frac{1}{2}\right)^2$.
\end{thm}

\section{Bias Lower Bounds}
\label{sec:sec4}

\subsection{$f$-RAC with PR}

We start by studying the $f$-RAC with PR. The following result is based on Ambainis \emph{et al.}~\cite{ambainis1999dense} and uses a procedure reminiscent of the concatenation idea from~\cite{pawlowski2009information,pawlowski2010entanglement,tuanuasescu2020optimal}: the initial string is broken in blocks, which in turn are encoded using the code from~\cite{cohen1983nonconstructive}. First we state a slightly modified version of Newman's Theorem~\cite{newman91} (see also~\cite[Theorem~3.14]{kushilevitz97} and~\cite[Theorem~3.5]{rao2020communication}) which is going to be useful to us.
\begin{thm}[\cite{newman91}]
    \label{thr:thr4.00}
    Let $\mathcal{E}(x,r)$ be an event depending on $x\in\{-1,1\}^{n}\times\{-1,1\}^n$ and $r\in R$ such that
    \begin{align*}
        \operatorname*{Pr}_{r\sim R}[\mathcal{E}(x,r)] \geq p
    \end{align*}
    for all $x\in\{-1,1\}^{n}\times\{-1,1\}^n$, with $p\in(0,1]$. Let $\delta\in(0,p]$. Then there is $R_0\subseteq R$ with size at most $n/\delta^2$ such that
    \begin{align*}
        \operatorname*{Pr}_{r\sim R_0}[\mathcal{E}(x,r)] \geq p - \delta
    \end{align*}
    holds for all $x\in\{-1,1\}^{n}\times\{-1,1\}^n$.
\end{thm}
\begin{thm}
    \label{thr:thr4.0}
    Let $\ell\in\mathbb{N}$, $\ell|n$, $m=\Omega(\ell\log{n})$ and $k=o(\sqrt{\ell})$. Let $f:\{-1,1\}^k\to\{-1,1\}$ be a Boolean function. For sufficiently large $n$ and $\ell$, there is an $n\overset{\varepsilon}{\mapsto}m$ $f$-RAC with PR and bias $\varepsilon \geq (1-o_n(1))\operatorname{Stab}_q[f]$ with
    \begin{align*}
        q \geq \sqrt{\frac{m}{n} - \frac{5\log_2{(n/\ell)}}{n/\ell}}.
    \end{align*}
\end{thm}
\begin{proof}
    Consider a code $C\subseteq\{-1,1\}^n$ such that, for every $x\in\{-1,1\}^n$, there is a $y\in C$ within Hamming distance $(1-p-\frac{1}{n})n$, with $p>1/2$ (the extra $1/n$ term will be used to counterbalance the decrease in probability from Newman's theorem). It is known~\cite{cohen1983nonconstructive} that there is such a code $C$ of size
    \begin{align*}
        \log_2{|C|} = \left(1 - H(p+1/n)\right)n + 2\log_2{n} \leq \left(1 - H(p)\right)n + 4\log_2{n}.
    \end{align*}
    Let $C(x)$ denote the closest codeword to $x$. Hence at least $(p+1/n)n$ out of $n$ bits of $C(x)$ are the same as $x$, and the probability over a uniformly random $i$ that $x_i = C(x)_i$ is at least $p+1/n$. 
    
    Let $\ell\in\mathbb{N}$ such that $\ell$ divides $n$. Our protocol involves breaking up $x\in\{-1,1\}^n$ into $\ell$ parts and encoding each part with the above code $C\subseteq\{-1,1\}^{n/\ell}$. Define the map
    \begin{align*}
        C^{(\ell)}(x) = C(x_1\dots x_{n/\ell})C(x_{n/\ell + 1}\dots x_{2n/\ell})\dots C(x_{(\ell-1)n/\ell+1}\dots x_n)
    \end{align*}
    that applies $C$ to the first $\ell$ bits of $x$, and to next $\ell$ bits of $x$ and so on. Hence the probability that $x_i = C^{(\ell)}(x)_i$ over a uniformly random $i$ is at least $p+\ell/n$. In order to consider this probability for \emph{every} bit instead of just an average over all bits, we employ the following randomization process. Let $r\in\{-1,1\}^n$ and $\pi\in\mathbb{S}_n$, both taken uniformly at random. Given $x\in\{-1,1\}^n$, denote $\pi(x) = x_{\pi(1)}x_{\pi(2)}\dots x_{\pi(n)}$. We define the encoding $C^{(\ell)}_{\pi,r}(x) = \pi^{-1}(C^{(\ell)}(\pi(x\cdot r))) \cdot r$, where $x\cdot r$ denotes the bit-wise product of $x$ and $r$. Let $\mathcal{E}_S$ be the event that all indices in $S\in\mathcal{S}_n^k$ are encoded in different codes $C$, i.e., in different blocks from $C^{(\ell)}$. There are $\ell$ blocks, each with $n/\ell$ elements. The probability that $k$ specific elements fall into $k$ different blocks is
    \begin{align}
        \label{eq:eq2.0}
        \operatorname*{Pr}_{S\sim\mathcal{S}_n^k}[\mathcal{E}_S] = \left(\frac{n}{\ell}\right)^k\frac{\binom{\ell}{k}}{\binom{n}{k}} = \prod_{j=1}^{k-1} \frac{1 - \frac{j}{\ell}}{1 - \frac{j}{n}} \geq \prod_{j=1}^{k-1} \left(1 - \frac{j}{\ell}\right) \overset{(a)}{\geq} 1 - \sum_{j=1}^{k-1}\frac{j}{\ell} = 1 - \frac{k(k-1)}{2\ell},
    \end{align}
    where inequality (a) can easily be proven by induction or the union bound.
    
    We shall first present a protocol using shared randomness, and at the end we shall transform it into a protocol with private randomness by using Newman's theorem. The protocol is the following. Select $r\in\{-1,1\}^n$ and $\pi\in\mathbb{S}_n$ uniformly at random. Encode $x$ as $C^{(\ell)}_{\pi,r}(x)$. To decode $f(x_S)$, first we check if all the indices of $S$ were encoded into different blocks. If no, the value for $f(x_S)$ is drawn uniformly at random. If yes, just consider $C^{(\ell)}_{\pi,r}(x)_S$ and evaluate $f(C^{(\ell)}_{\pi,r}(x)_S)$. Conditioned on the event $\mathcal{E}_S$ happening, the probability that $x_{S_i} = C^{(\ell)}(x)_{S_i}$ is at least $p+\ell/n$ independently for all $i\in[k]$, meaning that
    \begin{align*}
        \operatorname*{Pr}_{\pi,r}[f(x_S)=f(C^{(\ell)}_{\pi,r}(x)_S)|\mathcal{E}_S] \geq         \operatorname*{Pr}_{\substack{(x,y) \\ (q+\frac{2\ell}{n})\text{-correlated}}}[f(x)=f(y)] = \frac{1}{2} + \frac{1}{2}\operatorname{Stab}_{q+\frac{2\ell}{n}}[f],
    \end{align*}
    where $q := 2p-1$, and the inequality follows from monotonicity of the noise stability of $f$. With these considerations, the success probability of the protocol is
    \begin{align*}
       \operatorname*{Pr}_{\pi,r}[f(x_S)=f(C^{(\ell)}_{\pi,r}(x)_S)] &> 
        \frac{1}{2}\frac{k(k-1)}{2\ell} + \left(1 - \frac{k(k-1)}{2\ell}\right)\left(\frac{1}{2} + \frac{1}{2}\operatorname{Stab}_{q+\frac{2\ell}{n}}[f]\right)\\ 
        &\geq \frac{1}{2} + \frac{1}{2}\left(1 - \frac{k(k-1)}{2\ell}\right)\left(\operatorname{Stab}_{q}[f] + \operatorname{Stab}_{\frac{2\ell}{n}}[f]\right)\\
        &= \frac{1}{2} + \frac{1}{2}(1 - o_n(1))\left(\operatorname{Stab}_{q}[f] + \operatorname{Stab}_{\frac{2\ell}{n}}[f]\right),
    \end{align*}
    where we used that $k=o(\sqrt{\ell})$.
    
    We now transform the shared randomness into private randomness. By Newman's theorem (Theorem~\ref{thr:thr4.00}) there is a small set of permutation-string pairs (note that Alice's input is size $n$ bits and Bob's input is at most $n$ bits) with size
    \begin{align*}
        t \leq \frac{4n}{(1 - o_n(1))\operatorname{Stab}_{\frac{2\ell}{n}}[f]^2} \leq \frac{n^{2k+1}}{(1 - o_n(1))2^{2k-2}\ell^{2k}}
    \end{align*} 
    (we have used that $\operatorname{Stab}_a[f] \geq a^k$) such that $f(x_S)=f(C^{(\ell)}_{\pi,r}(x)_S)$ continues to hold with probability at least $\frac{1}{2} + \frac{1}{2}(1 - o_n(1))\operatorname{Stab}_q[f]$ if $\pi,r$ are chosen uniformly at random from this set. Hence the randomization can be encoded together with $x$ at the cost of a small overhead. The final protocol chooses $j\in[t]$ uniformly at random, encodes $x$ as $C^{(\ell)}_{\pi_j,r_j}(x)_S$ and then proceeds like the protocol with shared randomness. Fix $m = \log_2(t|C^{(\ell)}|)$. The result follows by using the first inequality from Theorem~\ref{thr:thr3.1} to observe that
    \begin{align*}
        \begin{multlined}[b][\textwidth]
            m \leq \big(1-H(p)\big)n + 4\ell\log_2\frac{n}{\ell} + \log_2{\frac{n^{2k+1}}{(1 - o_n(1))2^{2k-2}\ell^{2k}}} \leq q^2n + 4(1+o_{\ell}(1))\ell\log_2\frac{n}{\ell} \\
            \implies 2p - 1 \geq \sqrt{\frac{m}{n} - \frac{4(1+o_\ell(1))\log_2{(n/\ell)}}{n/\ell}}
        \end{multlined}
    \end{align*}
    for sufficiently large $n$, where we used $k = o(\sqrt{\ell})$ again.
\end{proof}

\begin{rem}
    The parameter $\ell$ in Theorem~\ref{thr:thr4.0} controls the number of encoding blocks in the protocol. By tweaking it, we can adjust the range of $m$ and $k$, e.g.\ if $\ell = \Theta(\log{n})$, then $m=\Omega(\log^2{n})$ and $k=o(\sqrt{\log{n}})$. If $\ell = \Theta(\sqrt{n})$, then $m=\Omega(\sqrt{n}\log{n})$ and $k = o(n^{1/4})$.
\end{rem}

In the protocol from Theorem~\ref{thr:thr4.0} we broke the initial string into $\ell$ different blocks and used $\ell$ different copies of $C$. This was done in order to guarantee the independence of the $C(x)_{S_i}$'s and hence analyse the influence of the code $C$ on the function $f$. Interestingly enough, for the special case of the Parity function this is not required and a single copy of $C$ can be used.
\begin{thm}
    \label{thr:thr5.2}
    Let $\operatorname{XOR}_k:\{-1,1\}^k\to\{-1,1\}$ be the Parity function and let $m=\Omega(k\log{n})$. There is an $n\overset{\varepsilon}{\mapsto} m$ $\operatorname{XOR}_k$-RAC with PR and bias 
    \begin{align}
        \label{eq:eq2.2}
        \varepsilon \geq \frac{1}{\binom{n}{k}}\mathcal{K}_{k,n}\left(\frac{n}{2} - \frac{n}{2}\sqrt{\frac{m}{n} - \frac{7k\log_2{n}}{n}}\right),
    \end{align}
    where $\mathcal{K}_{k,n}(x) = \sum_{j=0}^k (-1)^j \binom{x}{j}\binom{n-x}{k-j}$ is the Krawtchouk polynomial.
\end{thm}
\begin{proof}
    Consider the encoding $C_{\pi,r}(x) = \pi^{-1}(C(\pi(x\cdot r)))\cdot r$, where $C\subseteq\{-1,1\}^n$ is the code described in Theorem~\ref{thr:thr4.0}. Let $\delta n$ be the Hamming distance between $x$ and $C(x)$, with $\delta \leq 1 - p - 1/n$ by the properties of $C$. Then
    \begin{align*}
	    \operatorname*{Pr}_{S\sim\mathcal{S}_n^k}\left[\prod_{i=1}^k x_{S_i} = \prod_{i=1}^k C(x)_{S_i}\right] = \sum_{\ell=0}^{\lfloor \frac{k}{2}\rfloor} \frac{\binom{(1-\delta) n}{k-2\ell}\binom{\delta n}{2\ell}}{\binom{n}{k}} =
	    \frac{1}{2} + \frac{1}{2}\frac{\mathcal{K}_{k,n}(\delta n)}{\binom{n}{k}} \geq \frac{1}{2} + \frac{1}{2}\frac{\mathcal{K}_{k,n}((1-p)n)}{\binom{n}{k}} + \frac{1}{\binom{n}{k}},
    \end{align*}
    where we used $\sum_{\ell=0}^k \binom{(1-\delta)n}{k-\ell}\binom{\delta n}{\ell} = \binom{n}{k}$ on the second equality and $\mathcal{K}_{k,n}(\delta n) - \mathcal{K}_{k,n}(\delta n + 1) = 2$ on the final inequality, which can be obtained via the recurrence relation $\mathcal{K}_{k,n}(x) - \mathcal{K}_{k,n}(x-1) = \mathcal{K}_{k-1,n}(x) - \mathcal{K}_{k-1,n}(x-1)$ and $\mathcal{K}_{1,n}(x) = n - 2x$ (see e.g.~\cite{coleman2011krawtchouk}).
    
    By Newman's theorem (Theorem~\ref{thr:thr4.00}) there is a small set of permutation-string pairs with size $t=n\binom{n}{k}^2$ such that $\prod_{i=1}^k x_{S_i} = \prod_{i=1}^k C(x)_{S_i}$ continues to hold with bias at least $\mathcal{K}_{k,n}((1-p)n)/\binom{n}{k}$ for \emph{any} $x$ and $S$ if $\pi,r$ are chosen uniformly at random from this set.
    
    Our protocol is the following. Select $j\in[t]$ uniformly at random. Encode $x$ as $C_{\pi_j,r_j}(x)$. To decode $\prod_{i=1}^k x_{S_i}$, just consider $C_{\pi_j,r_j}(x)_S$ and evaluate $\prod_{i=1}^k C_{\pi_j,r_j}(x)_{S_i}$. Now fix $m = \log_2(t|C|)$. The result follows by using the first inequality from Theorem~\ref{thr:thr3.1} to observe that
    \begin{equation*}
        m \leq \big(1-H(p)\big)n + 5\log_2{n} + 2\log_2{\binom{n}{k}}
            \implies 1 - p \leq \frac{1}{2} - \frac{1}{2}\sqrt{\frac{m}{n} - \frac{7k\log_2{n}}{n}}. \qedhere
    \end{equation*}
\end{proof}
\begin{rem}
    Since $\mathcal{K}_{1,n}(x) = n - 2x$, we note that, for $k=1$, Eq.~(\ref{eq:eq2.2}) reduces to $\varepsilon \geq \sqrt{\frac{m}{n} - \frac{7\log_2{n}}{n}}$, which is the result from Ambainis \emph{et al.}~\cite{ambainis1999dense} (see Theorem~\ref{thr:thr2.1a}).
\end{rem}
\begin{rem}
    If $k=O(1)$, then the Krawtchouk polynomial has the asymptotic limit as $n\to\infty$ of $\mathcal{K}_{k,n}(x)\sim \frac{2^k n^k}{k!}\left(\frac{1}{2} - \frac{x}{n}\right)^k$~\cite[Eq.~(29)]{dominici2008asymptotic}, thus the bias from Theorem~\ref{thr:thr5.2} has the asymptotic limit
    \begin{align*}
        \varepsilon \sim \frac{n^k}{k!\binom{n}{k}}\left(\frac{m}{n} - \frac{7k\log_2{n}}{n}\right)^{k/2} \geq \left(\frac{m}{n} - \frac{7k\log_2{n}}{n}\right)^{k/2}.
    \end{align*}
    Note that this result is very similar to the one that would follow from Theorem~\ref{thr:thr4.0}, but slightly tighter (without the $\ell$ parameter and the multiplicative constant $1-o_n(1)$).
\end{rem}

\subsection{$f$-RAC with SR}

There is a lower limit of $m=\Omega(\log{n})$ on the number of encoded bits in Theorem~\ref{thr:thr4.0}. It is possible to go below this limit by using SR: the blocks are now encoded via the $n\overset{\varepsilon}{\mapsto}1$ RAC with SR from Theorem~\ref{thr:thr2.1} instead of the code $C$.
\begin{thm}
    \label{thr:thr4.1}
    Let $m|n$ and $k=o(\sqrt{m})$. Let $f:\{-1,1\}^k\to\{-1,1\}$ be a Boolean function. There is an $n\overset{\varepsilon}{\mapsto}m$ $f$-RAC with SR and bias $\varepsilon \geq (1-o_n(1))\operatorname{Stab}_q[f]$ with
    \begin{align*}
        q = \frac{2}{2^{n/m}}\binom{n/m - 1}{\lfloor \frac{n/m - 1}{2}\rfloor} \geq \begin{cases} \sqrt{\frac{2m}{\pi n}} - O\left(\frac{1}{(n/m)^{3/2}}\right),\\
        \sqrt{\frac{m}{2n}}.
        \end{cases}
    \end{align*}
\end{thm}
\begin{proof}
    Consider the RAC with SR from Theorem~\ref{thr:thr2.1}. Our protocol is the following. For the encoding of $x\in\{-1,1\}^n$, its $n$ bits are randomly divided into $m$ sets $T_1,\dots,T_m$, each with $n/m$ elements. Each set is encoded into $1$ bit with the $n/m\overset{\epsilon}{\mapsto} 1$ RAC, and the encoded string is $E(x)\in\{-1,1\}^m$. For decoding, one checks with the help of SR if all the $k$ indices in $S$ were encoded into different sets. If no, the value for $f(x_S)$ is drawn uniformly at random. If yes, let $D_i:\{-1,1\}\to\{-1,1\}$ be the decoding function for set $T_i$, which corresponds to bit $E(x)_i$. Then $z_{\ell_i} := D_{\ell_i}(E(x)_{\ell_i})$ is the decoded $x_{S_i}$, where, for all $i\in[k]$, $\ell_i\in[m]$ is such that $S_i \in T_{\ell_i}$.\footnote{The decoding map of the RAC from Theorem~\ref{thr:thr2.1} (see~\cite[Theorem 2]{ambainis2008quantum}) is just the identity map, so $z_{\ell_i} = E(x)_{\ell_i}$.} Write $z_T = z_{\ell_1}\dots z_{\ell_k}$. We output $f(z_T)$ for $f(x_S)$. 

    Let $\mathcal{E}_S$ be the event that all indices in $S\in\mathcal{S}_n^k$ are encoded in different sets. Similarly to Eq.~(\ref{eq:eq2.0}),
    \begin{align*}
        \operatorname*{Pr}_{S\sim\mathcal{S}_n^k}[\mathcal{E}_S] \geq 1 - \frac{k(k-1)}{2m}.
    \end{align*}
    The bias of correctly recovering any of the $n/m$ encoded bits is $q = \frac{2}{2^{n/m}}\binom{n/m - 1}{\lfloor \frac{n/m - 1}{2}\rfloor}$ by Theorem~\ref{thr:thr2.1}. Conditioning on $\mathcal{E}_S$ happening, we see that $x_S$ and $y_T$ are $q$-correlated according to Definition~\ref{def:def1}. Therefore
    %
    %
    \begin{align*}
        \operatorname*{Pr}_{T_1,\dots,T_m}[f(x_S)=f(z_T)|\mathcal{E}_S] =
        \operatorname*{Pr}_{\substack{(x,y) \\q\text{-correlated}}}[f(x)=f(y)] = \frac{1}{2} + \frac{1}{2}\operatorname{Stab}_q[f],
    \end{align*}
    where we used an input randomization via SR. With these considerations, the success probability of the protocol is
    \begin{align*}
        \operatorname*{Pr}_{T_1,\dots,T_m}[f(x_S)=f(z_T)] &\geq 
        \frac{1}{2}\frac{k(k-1)}{2m} + \left(1 - \frac{k(k-1)}{2m}\right)\left(\frac{1}{2} + \frac{1}{2}\operatorname{Stab}_q[f]\right)\\ 
        &= \frac{1}{2} + \frac{1}{2}\left(1 - \frac{k(k-1)}{2m}\right)\operatorname{Stab}_q[f]\\
        &= \frac{1}{2} + \frac{1}{2}(1 - o_n(1))\operatorname{Stab}_q[f],
    \end{align*}
    using that $k=o(\sqrt{m})$, from where the result follows by noticing that
    \begin{align*}
        q = \frac{2}{2^{n/m}}\binom{n/m - 1}{\lfloor \frac{n/m - 1}{2}\rfloor} \geq \begin{cases} \sqrt{\frac{2m}{\pi n}} - O\left(\frac{1}{(n/m)^{3/2}}\right),\\
        \sqrt{\frac{m}{2n}},
        \end{cases}
    \end{align*}
    with $n! = \sqrt{2\pi n}(n/e)^n(1 + O(1/n))$ and $\binom{2n}{n} \geq 2^{2n}/(2\sqrt{n})$~\cite[Chapter 10, Lemma 7]{macwilliams1977theory}.
\end{proof}

\begin{rem}
    It is possible to use Newman's theorem in the above theorem in order to transform SR into PR, but then $\Omega(\log{n})$ encoding bits would need to be used to encode the randomization procedure, thus leading to $m=\Omega(\log{n})$. Moreover, the final $f$-RAC would have worse bias compared to the one from Theorem~\ref{thr:thr4.0}.
\end{rem}
\begin{rem}
    The requirement $m|n$ can be dropped by adding extra bits into $x\in\{-1,1\}^n$ until $m|n'$, where $n'$ is the final number of bits.
\end{rem} 
\begin{rem}
    For $m=n/2$ or $m=n/3$, we have $\frac{2}{2^{n/m}}\binom{n/m - 1}{\lfloor \frac{n/m - 1}{2}\rfloor} = \frac{1}{2}$, so the resulting biases have $q_{m=n/2} = q_{m=n/3} = 1/2$. Also, by using induction, $\frac{2}{2^{n}}\binom{n - 1}{\lfloor \frac{n - 1}{2}\rfloor} \leq \frac{1}{2}$ for any $n\geq 2$. 
    %
\end{rem} 

\subsection{$f$-QRAC}

Exactly the same procedure from Theorem~\ref{thr:thr4.1} holds for $f$-QRACs if we replace the $n/m\mapsto1$ RAC from Theorem~\ref{thr:thr2.0} with the $n/m\mapsto1$ QRAC from Theorem~\ref{thr:thr2.2} when encoding the sets $T_1,\dots,T_m$.

\begin{thm}
    \label{thr:thr4.2}
    Let $m|n$ and $k=o(\sqrt{m})$. Let $f:\{-1,1\}^k\to\{-1,1\}$ be a Boolean function. There is an $n\overset{\varepsilon}{\mapsto}m$ $f$-QRAC with SR and bias $\varepsilon \geq (1-o_n(1))\operatorname{Stab}_q[f]$ with
    \begin{align*}
        q = \sqrt{\frac{8m}{3\pi n}} + O\left(\frac{1}{(n/m)^{3/2}}\right).
    \end{align*}
\end{thm}
\begin{proof}
    Replace the $n/m\mapsto 1$ RAC in the proof of Theorem~\ref{thr:thr4.1} with the $n/m\overset{\epsilon}{\mapsto}1$ QRAC from Theorem~\ref{thr:thr2.2} with bias $\epsilon = \sqrt{\frac{8m}{3\pi n}} + O\left(\frac{1}{(n/m)^{3/2}}\right)$.
\end{proof}

\begin{rem}
    If $m=n/2$ or $m=n/3$, the usual (and optimal) $2\overset{1/\sqrt{2}}{\longmapsto}1$ QRAC or $3\overset{1/\sqrt{3}}{\longmapsto}1$ QRAC with PR can be used, respectively. The resulting biases have $q_{m=n/2} = 1/\sqrt{2}$ and $q_{m=n/3} = 1/\sqrt{3}$.
\end{rem} 

\subsection{$f$-EARAC}
\label{sec:sec5.3}

The same protocol can also be used for $f$-EARACs, now with the $n/m\mapsto1$ EARAC from Theorem~\ref{thr:thr2.3}.

\begin{thm}
    \label{thr:thr4.5}
    Let $m|n$ and $k=o(\sqrt{m})$. Let $f:\{-1,1\}^k\to\{-1,1\}$ be a Boolean function. There is an $n\overset{\varepsilon}{\mapsto}m$ $f$-EARAC with bias $\varepsilon \geq (1-o_n(1))\operatorname{Stab}_q[f]$ and
    \begin{align*}
        q = \sqrt{\frac{m}{n}}.
    \end{align*}
\end{thm}
\begin{proof}
    Replace the $n/m\mapsto 1$ RAC in the proof of Theorem~\ref{thr:thr4.1} with the $n/m\overset{\epsilon}{\mapsto} 1$ EARAC from Theorem~\ref{thr:thr2.3} with bias $\epsilon = \sqrt{m/n}$.
\end{proof}
\begin{rem}
    We could also define an entanglement-assisted $f$-QRAC ($f$-EAQRAC) similarly to Definition~\ref{def:def10}, i.e., as an $f$-QRAC with SR where both parties share an unlimited amount of entanglement. Due to super-dense coding and teleportation, an $n\overset{\varepsilon}{\mapsto}m$ $f$-EAQRAC is equivalent to an $n\overset{\varepsilon}{\mapsto}2m$ $f$-EARAC, meaning that there is an $n\overset{\varepsilon}{\mapsto}m$ $f$-EAQRAC with $k=o(\sqrt{m})$ and bias $\varepsilon \geq (1-o_n(1))\operatorname{Stab}_q[f]$ with $q=\sqrt{\frac{2m}{n}}$.
\end{rem}


If $f:\{-1,1\}\to\{-1,1\}$ is $f(x) = x$, i.e., when considering the usual $n\mapsto m$ EARAC, the above result tells us that the success probability is just
\begin{align}
    \label{eq:eq3.1}
    p = \frac{1}{2} + \frac{1}{2}\sqrt{\frac{m}{n}}.
\end{align}
Moreover, we note that the $n\mapsto m$ EARAC is formed by a grouping of $n/m\mapsto 1$ EARACs, such that the $m$ outcomes, i.e., the bits of the encoding message $E(x)\in\{-1,1\}^m$, are all independent of each other. More precisely, for each $i\in[n]$, there is a unique $j\in[m]$ such that $\operatorname{Pr}[x_i|E(x)] = \operatorname{Pr}[x_i|E(x)_j]$. Under this assumption, we can prove that Eq.~(\ref{eq:eq3.1}) is optimal by the optimality of its parts. Consider breaking the initial string $x\in\{-1,1\}^n$ into $m$ blocks, the $i$-th block containing $r_i$ elements. Then the success probability of the resulting $n\mapsto m$ EARAC is
\begin{align*}
    \sum_{i=1}^m \frac{r_i}{n} \left(\frac{1}{2} + \frac{1}{2\sqrt{r_i}}\right) = \frac{1}{2} + \frac{1}{2n}\sum_{i=1}^m \sqrt{r_i},
\end{align*}
which is maximized by taking $r_i = n/m$ for all $i\in[m]$. Since the $n/m\mapsto 1$ EARACs are optimal by Theorem~\ref{thr:thr2.3}, so is Eq.~(\ref{eq:eq3.1}) (under the assumption that the $n\mapsto m$ EARAC is formed by $m$ independent EARACs on $1$ encoding bit).

We can use Theorem~\ref{thr:thr2.0} to obtain the following Corollary from Theorems~\ref{thr:thr4.0},~\ref{thr:thr4.1},~\ref{thr:thr4.2} and~\ref{thr:thr4.5}.

\begin{cor}
    \label{cor:cor2}
    Let $f:\{-1,1\}^k\to\{-1,1\}$ be a Boolean function with \emph{pure high degree} $\mathcal{h} = \min\{|S|:\widehat{f}(S)\neq 0\}$. Let $W^j[f] = \sum_{\substack{S\subseteq[k]\\|S|=j}} \widehat{f}(S)^2$.
    \begin{enumerate}[(a)]
        \item Let $\ell\in\mathbb{N}$, $m=\Omega(\ell\log{n})$ and $k=o(\sqrt{\ell})$. There is an $n\overset{\varepsilon}{\mapsto}m$ $f$-RAC with PR and bias $\varepsilon \geq (1-o_n(1))W^\mathcal{h}[f]\left(\frac{m}{n} - \frac{5\log_2{(n/\ell)}}{n/\ell}\right)^{\mathcal{h}/2}$.
        \item Let $k=o(\sqrt{m})$. There is an $n\overset{\varepsilon}{\mapsto}m$ $f$-RAC with SR and bias $\varepsilon \geq (1-o_n(1)) W^\mathcal{h}[f]\left(\frac{m}{2 n}\right)^{\mathcal{h}/2}$. 
        \item Let $k=o(\sqrt{m})$. There is an $n\overset{\varepsilon}{\mapsto}m$ $f$-QRAC with SR and bias $\varepsilon \geq (1-o_n(1))W^\mathcal{h}[f] \left(\frac{8m}{3\pi n}\right)^{\mathcal{h}/2}$.
        \item Let $k=o(\sqrt{m})$. There is an $n\overset{\varepsilon}{\mapsto}m$ $f$-EARAC with bias $\varepsilon \geq (1-o_n(1))W^\mathcal{h}[f] \left(\frac{m}{n}\right)^{\mathcal{h}/2}$.
    \end{enumerate}
\end{cor}

\subsection{$f$-PRRAC}

We now present a protocol for the $f$-PRRAC, based on reducing the problem to the standard random access code setting, and then using a protocol defined in~\cite{pawlowski2009information}. This protocol was used to show the violation of information causality by means of a pyramid of Popescu-Rohrlich boxes and nesting van Dam's protocol~\cite{van2013implausible}, which allows us to decode the value of $f(x_S)$ for any $S\in\mathcal{S}_n^k$ with just one encoded bit. This procedure of pyramiding and nesting was also used in the context of EARACs in~\cite{pawlowski2010entanglement,tuanuasescu2020optimal} under the name of concatenation.

\begin{thm}
    Let $f:\{-1,1\}^k\to\{-1,1\}$ be a Boolean function. For any $n\in\mathbb{N}$, there is an $n\overset{1}{\mapsto}1$ $f$-PRRAC.
\end{thm}
\begin{proof}
    To ease the notation, we shall use $\{0,1\}$ instead of $\{-1,1\}$ during the proof. We shall also name the encoding and decoding parties Alice and Bob, respectively, and refer to a Popescu-Rohrlich box as PR-box. Let\footnote{If $f$ is symmetric, the size $t$ can be decreased to $\binom{n}{k}$ by ignoring all the redundant permutations of the $k$-sets.} $t := |\mathcal{S}_n^k| = k!\binom{n}{k}$ and define the string $a\in\{0,1\}^t$ as $a_S := f(x_S)$, where $\mathcal{S}_n^k$ is arranged in lexicographic order.\footnote{Here we use $a_S$ to denote a single bit of $a$, whereas $x_S$ denotes a subsequence of $x$.} Bob is interested in bit $a_S$, whose index position can be described by a $t$-bit string $b = \sum_{i=0}^{t-1}b_i 2^i$, i.e., the considered function is $g_t(a,b) := a_b$. The remainder of the argument is the same as the protocol of~\cite{pawlowski2009information}, but we include the details for completeness. For $t = 1$ we have that
    \begin{align*}
        g_t((a_0,a_1),b_0) = a_0\oplus b_0(a_0\oplus a_1).
    \end{align*}
    Alice inputs $a_0\oplus a_1$ into a PR-box, while Bob inputs $b_0$. Alice obtains the output $A$ and sends the message $y=a_0\oplus A$ to Bob, who can obtain $y\oplus B=a_b$ using his output $B$, since, by the PR-box property, $A\oplus B = b_0(a_0\oplus a_1)$.
    
    For $t > 1$, write $a = a'a''$, where $a' = a_0\dots a_{t/2 - 1}\in\{0,1\}^{t/2}$ and $a'' = a_{t/2}\dots a_{t -1}\in\{0,1\}^{t/2}$. Then one can show that
    \begin{align*}
        g_t(a,b) = g_{t-1}(a',b')\oplus b_{t-1}\left(g_{t-1}(a',b')\oplus g_{t-1}(a'',b')\right),
    \end{align*}
    where $b' = b_0\dots b_{t-2}\in\{0,1\}^{t-1}$. Therefore we can construct a recursive protocol in $t$, which will encompass all values of $n$. The protocol uses a pyramid of $2^t-1$ Popescu-Rohrlich boxes placed on $t$ levels. The case $t = 1$ was explained above. For $t > 1$, Alice and Bob use the protocol on inputs $(a',b')$ and $(a'',b')$, which involves $2^{t/2} - 1$ PR-boxes in each one. Alice's outputs of each protocol are $y'$ and $y''$, which she inputs into the last PR-box, similarly to the case $t=1$, as $y'\oplus y''$, while Bob inputs $b_{t-1}$. Given Alice's final output $A$, she sends $y=y'\oplus A$ to Bob, who uses his output $B_{t-1}$ to obtain $y'\oplus b_{t-1}(y'\oplus y'')$. If $b_{t-1} = 0$, he gets $y'$, otherwise, if $b_{t-1}=1$, he gets $y''$. With these, he can recursively go up the pyramid based on the protocol for $t-1$ bits, which tells him which boxes to read. Looking at the binary decomposition of $b$, Bob goes $(t-r)$ times to the left bit, and $r$ times to the right bit, where $r=\sum_{i=0}^{t-1}b_i$. His final output will be $y\oplus B_0\oplus\cdots\oplus B_{t-1}$, where $B_j$ is the output for the PR-box that Bob uses at level $j$. Bob will only need the outputs of $t$ PR-boxes, while Alice uses $2^{t} - 1$ PR-boxes in total.
\end{proof}

\section{Bias Upper Bounds}
\label{sec:sec5}

In order to prove an upper bound on the bias of any $n\overset{\varepsilon}{\mapsto}m$ $f$-QRAC with SR, we shall use the following equivalent version of Definition~\ref{def:def9}, which comes from input randomization, i.e., from considering the average success probability over the inputs, and from the following fact.
\begin{fact}[\cite{helstrom1976quantum}]
    \label{lem:lem3.5.c3}
    Let $\rho$ be an unknown state picked from the set $\{\rho_0,\rho_1\}$ with probability $p$ and $1-p$, respectively.
    The optimal success probability of predicting which state it is by a POVM is
    \begin{align*}
        \frac{1}{2} + \frac{1}{2}\|p\rho_0 - (1-p)\rho_1\|_{\operatorname{tr}}.
    \end{align*}
\end{fact}
\begin{defn}
    Let $f:\{-1,1\}^k\to\{-1,1\}$ be a Boolean function. An $n\overset{\varepsilon}{\mapsto}m$ $f$-QRAC with SR is a map $\rho:\{-1,1\}^n\to\mathbb{C}^{2^m\times 2^m}$ that assigns an $m$-qubit density matrix $\rho(x)$ to every $x\in\{-1,1\}^n$ and has the property that
		\begin{align*}
			\operatorname*{\mathbb{E}}_{S\sim \mathcal{S}_n^k}\left[\frac{1}{2^n}\left\|\sum_{x\in\{-1,1\}^n} f(x_S)\rho(x)\right\|_{\operatorname{tr}}\right] \geq \varepsilon.
		\end{align*}
\end{defn}

\begin{thm}
    \label{thr:thr4.4}
    Let $f:\{-1,1\}^k\to\{-1,1\}$ be a Boolean function. For any $n\overset{\varepsilon}{\mapsto}m$ $f$-QRAC with SR and $k=o(n)$ the following holds: for any $\eta > 2\ln{2}$ there is a constant $C_\eta$ such that
    \begin{align}
        \label{eq:eq4.2}
        \varepsilon \leq C_\eta\sum_{\ell=0}^k L_{1,\ell}(f)\left(\frac{\eta m}{n}\right)^{\ell/2},
    \end{align}
    where $L_{1,\ell}(f) = \sum_{\substack{T\subseteq[k]\\ |T| = \ell}}\big|\widehat{f}(T)\big|$ is the $1$-norm of the $\ell$-th level of the Fourier transform of $f$.
\end{thm}
\begin{proof}
    We start by writing the following.
	\begin{align*}
		\frac{1}{2^n}\left\|\sum_{x\in\{-1,1\}^n} f(x_S)\rho(x)\right\|_{\operatorname{tr}} &=
		\frac{1}{2^n}\left\|\sum_{V\subseteq[n]}\sum_{T\subseteq[k]}  \widehat{f}(T)\widehat{\rho}(V)\sum_{x\in\{-1,1\}^n} \chi_V(x)\chi_T(x_S)\right\|_{\operatorname{tr}}\\
		&= \left\|\sum_{T\subseteq[k]}  \widehat{f}(T)\widehat{\rho}(S_{T})\right\|_{\operatorname{tr}} \\
		&\leq \sum_{T\subseteq[k]}\big|\widehat{f}(T)\big|\|\widehat{\rho}(S_{T})\|_{\operatorname{tr}},
	\end{align*}
	where $S_{T} = \{S_i ~|~ i\in T\}$. 
	Then
	\begin{align*}
		\varepsilon \leq \sum_{\ell = 0}^k\sum_{\substack{T\subseteq[k]\\ |T| = \ell}}\big|\widehat{f}(T)\big|\operatorname*{\mathbb{E}}_{S\sim \mathcal{S}_n^k}\big[\|\widehat{\rho}(S_{T})\|_{\operatorname{tr}}\big],
	\end{align*}
	but, for a given $T\subseteq[k]$ with $|T| = \ell$,
	\begin{align*}
		\operatorname*{\mathbb{E}}_{S\sim \mathcal{S}_n^k}\big[\|\widehat{\rho}(S_{T})\|_{\operatorname{tr}}\big] & \frac{1}{k!\binom{n}{k}}\sum_{S\in \mathcal{S}_n^k}\|\widehat{\rho}(S_{T})\|_{\operatorname{tr}} 
		= \frac{\ell!(k-\ell)!}{k!\binom{n}{k}}\binom{n-\ell}{k-\ell}\sum_{S\in \binom{[n]}{\ell}}\|\widehat{\rho}(S)\|_{\operatorname{tr}}
		= \operatorname*{\mathbb{E}}_{S\sim \binom{[n]}{\ell}}\big[\|\widehat{\rho}(S)\|_{\operatorname{tr}}\big],
	\end{align*}
	and thus, using Jensen's inequality,
	\begin{align*}
		\varepsilon \leq \sum_{\ell = 0}^k\sum_{\substack{T\subseteq[k]\\ |T| = \ell}}\big|\widehat{f}(T)\big|\operatorname*{\mathbb{E}}_{S\sim \binom{[n]}{\ell}}\big[\|\widehat{\rho}(S)\|_{\operatorname{tr}}\big] 
		\leq \sum_{\ell = 0}^k L_{1,\ell}(f) \sqrt{\operatorname*{\mathbb{E}}_{S\sim \binom{[n]}{\ell}}\big[\|\widehat{\rho}(S)\|^2_{\operatorname{tr}}\big]}.
	\end{align*}
	We now use Theorem~\ref{thm:thm2.2.c1} with $\delta = \frac{\ell}{(2\ln{2})m}$, taking only the sum on $S$ with $|S| = \ell$,
	\begin{align*}
		\operatorname*{\mathbb{E}}_{S\sim \binom{[n]}{\ell}}\big[\|\widehat{\rho}(S)\|^2_{\operatorname{tr}}\big] \leq 2^{2\delta m}\delta^{-\ell}\binom{n}{\ell}^{-1} = \left(\frac{(2e\ln{2})m}{\ell}\right)^\ell \binom{n}{\ell}^{-1},
	\end{align*}
	to finally obtain
	\begin{align*}
		\varepsilon \leq \sum_{\ell = 0}^k \binom{n}{\ell}^{-1/2} L_{1,\ell}(f) \left(\frac{(2e\ln{2})m}{\ell}\right)^{\ell/2}.
	\end{align*}
	From here we can use Stirling's approximation $n! = \Theta(\sqrt{n}(n/e)^n)$ to obtain
	\begin{align*}
		\binom{n}{\ell} =\Theta\left(\sqrt{\frac{n}{\ell(n-\ell)}}\left(\frac{n}{\ell}\right)^\ell \left(1 + \frac{\ell}{n-\ell}\right)^{n-\ell}\right).
	\end{align*}
	We use the fact that for large enough $n/\ell$ we have $(1 + \ell/(n-\ell))^{(n-\ell)/\ell} > (2e\ln{2})/\eta$, where $\eta > 2\ln{2}$, and that the factor $\sqrt{n/\ell(n-\ell)} \geq \sqrt{1/k}$ can be absorbed by this approximation. Then there is a constant $C_\eta$ such that Eq.~(\ref{eq:eq4.2}) holds.
\end{proof}

A few different bounds can be obtained from the above theorem, some with a clearer meaning.

\begin{cor}
    \label{cor:cor1}
    Let $f:\{-1,1\}^k\to\{-1,1\}$ be a Boolean function. Let $r\in[0,1]$. For any $n\overset{\varepsilon}{\mapsto}m$ $f$-QRAC with SR and $k=o(n)$ the following holds: for any $\eta > 2\ln{2}$ there is a constant $C_\eta$ such that
    \begin{subnumcases}{\varepsilon \leq}
            C_\eta \sqrt{\operatorname{Stab}_{q^{2r}}[f]\sum_{S\in\operatorname{supp}(\widehat{f})}q^{2(1-r)|S|}}, \label{eq:eq4.2a}\\
            C_\eta \hat{\|}f\hat{\|}_1 \left(\frac{\eta m}{n}\right)^{\mathcal{h}/2},\label{eq:eq4.2b}\\
            C_\eta 2^{\operatorname{deg}(f)-1}\operatorname{Stab}_q[f],\label{eq:eq4.2c}
    \end{subnumcases}
    where $q=\sqrt{\frac{\eta m}{n}}$, $\operatorname{supp}(\widehat{f}) = \{S\subseteq[k]~|~\widehat{f}(S) \neq 0\}$ is the support of $f$, $\mathcal{h} = \min\{|S|:\widehat{f}(S) \neq 0 \}$ is the pure high degree of $f$, $\operatorname{deg}(f) = \max\{|S|:\widehat{f}(S)\neq 0\}$ is the degree of $f$ and $\hat{\|}f\hat{\|}_1 = \sum_{S\subseteq[k]} |\widehat{f}(S)|$ is the Fourier $1$-norm of $f$.
\end{cor}
\begin{proof}
    From Theorem~\ref{thr:thr4.4} we know that for any $\eta > 2\ln{2}$ there is a constant $C_\eta$ such that
    \begin{align}
        \label{eq:eq4.5}
        \varepsilon \leq C_\eta\sum_{\ell=\mathcal{0}}^k L_{1,\ell}(f)\left(\frac{\eta m}{n}\right)^{\ell/2}.
    \end{align}
    
    There are a couple of ways to bound the above quantity. We start by proving Eq.~(\ref{eq:eq4.2a}). Define $g:\{-1,1\}^k\to\mathbb{R}$, $g=\sum_{S\in\operatorname{supp}(\widehat{f})} \operatorname{sgn}(\widehat{f}(S)) \chi_S$. Let $\operatorname{T}_q$ be the noise operator with parameter $q=\sqrt{\frac{\eta m}{n}}$. Let $r,s\in[0,1]$ be such that $r+s=1$. By Cauchy-Schwarz,
	\begin{align*}
	    \left(\sum_{S\subseteq[k]} q^{|S|} |\widehat{f}(S)|\right)^2 &= |\langle \operatorname{T}_{q^r}f, \operatorname{T}_{q^s} g\rangle|^2 
	    \leq \langle \operatorname{T}_{q^r}f, \operatorname{T}_{q^r}f\rangle \langle \operatorname{T}_{q^s}h, \operatorname{T}_{q^s}g\rangle
	    = \operatorname{Stab}_{q^{2r}}[f]\sum_{S\in\operatorname{supp}(\widehat{f})}q^{2s|S|}.
	\end{align*} 
	By plugging the above equation into Eq.~(\ref{eq:eq4.5}), Eq.~(\ref{eq:eq4.2a}) follows.
	
    Eqs.~(\ref{eq:eq4.2b}) and~(\ref{eq:eq4.2c}) follow, respectively, by
	\begin{align*}
		\varepsilon \leq C_\eta\sum_{\ell=\mathcal{0}}^k L_{1,\ell}(f)\left(\frac{\eta m}{n}\right)^{\ell/2} \leq C_\eta\hat{\|}f\hat{\|}_1\left(\frac{\eta m}{n}\right)^{\mathcal{h}/2}
	\end{align*}
	and
	\begin{align}
	    \label{eq:eq4.7}
        \varepsilon \leq C_\eta\sum_{\ell=\mathcal{0}}^k L_{1,\ell}(f)\left(\frac{\eta m}{n}\right)^{\ell/2}
        \leq C_\eta\sum_{\ell=\mathcal{0}}^k 2^{\operatorname{deg}(f)-1}W^\ell[f]\left(\frac{\eta m}{n}\right)^{\ell/2}
        = C_\eta 2^{\operatorname{deg}(f)-1}\operatorname{Stab}_q[f],
	\end{align}
	where we used that $f$'s Fourier spectrum is $2^{1-\operatorname{deg}(f)}$-granular in Eq.~(\ref{eq:eq4.7}), i.e., $\widehat{f}(S)$ is an integer multiple of $2^{1-\operatorname{deg}(f)}$ for all $S\subseteq[k]$ (see \cite[Exercise~1.11]{o2014analysis}).
\end{proof}

Corollary~\ref{cor:cor1} helps with the comparison between the bias upper bound and the bias lower bounds for $f$-RAC and $f$-QRAC (Theorems~\ref{thr:thr4.0},~\ref{thr:thr4.1} and~\ref{thr:thr4.2}). By taking $\deg(f)$ as constant, we see that Eq.~(\ref{eq:eq4.2c}) matches the bias lower bounds in terms of the noise stability up to an overall multiplicative constant and up to the multiplicative constant $\sqrt{\eta}$ in the parameter $q$ in $\operatorname{Stab}_q[f]$. Another comparison is between Eq.~(\ref{eq:eq4.2b}) and Corollary~\ref{cor:cor2} in terms of the pure high degree of $f$. Again both bounds match up to a global multiplicative constant and up to the constant $\eta$. We conjecture that the constant $\eta$ can be dropped from all these bounds with a more careful analysis.


\section{Conclusions}

In this paper we proposed a simple generalization of the concept of random access to recovering the value of a given Boolean function on any subset of fixed size of the initial bits. This generalization was made assuming different resources as encoding maps, i.e., encoding the initial string into bits or qubits, and different auxiliary resources, e.g.\ private and shared randomness, shared entanglement and Popescu-Rohrlich boxes. Given the lower bounds from our protocols, it seems reasonable to assume that the bias $\operatorname{Stab}_q[f]$ with $q=\sqrt{\frac{m}{n}}$ is, if not optimal, at least close to optimal. The case with the weakest resources, the $n\mapsto m$ $f$-RAC with PR, already achieves such bias up to an additive term $O((\log{n/\ell})/(n/\ell))$ in the parameter $q$. For more general values of $m=O(\log{n})$, the use of quantum resources progressively improves $q$: from $q\approx \sqrt{\frac{2m}{\pi n}}$ using encoding bits and SR to $q\geq \sqrt{\frac{8m}{3\pi n}}$ using encoding qubits and SR and finally to $q=\sqrt{\frac{m}{n}}$ using encoding bits and shared entanglement. Such an improvement offered by quantum resources is relatively modest, specially when compared to stronger-than-quantum resources like Popescu-Rohrlich boxes, which allows the recovery of $f(x_S)$ with certainty for any $S$.

On the other hand, the techniques from Fourier analysis lead to bias upper bounds that match our bias lower bounds up to a global multiplicative constant and a factor $\sqrt{\eta} \approx \sqrt{2\ln{2}}$ in the parameter $q$. We conjecture that such upper bounds can be improved and the factor $\eta$ dropped. Moreover, the upper bounds apply only to $f$-QRACs with SR, therefore not including $f$-EARACs. The understanding of EARACs is still limited, and even though we obtained an upper bound by making an independence assumption, a general upper bound for the case $m>1$ is yet unknown. 

\subsection*{Acknowledgements}

We thank Ronald de Wolf for suggesting the block-encoding scheme, originally with the $2,3\mapsto 1$ QRACs, for pointing out Refs.~\cite{MR2115303,klauck2007one} and for feedback on the manuscript. We thank M\'{a}t\'{e} Farkas and Mark Howard for pointing out Refs.~\cite{aguilar2018certifying,farkas2020self,farkas2019self,tavakoli2015quantum} and~\cite{emeriau2020quantum}, respectively. We acknowledge support from the QuantERA ERA-NET Cofund in Quantum Technologies implemented within the European Union’s Horizon 2020 Programme (QuantAlgo project) and EPSRC grants EP/R043957/1 and EP/T001062/1. This project has received funding from the European Research Council (ERC) under the European Union’s Horizon 2020 research and innovation programme (grant agreement No.\ 817581). JFD was supported by the Bristol Quantum Engineering Centre for Doctoral Training, EPSRC Grant No.\ EP/L015730/1.

\bibliographystyle{plainnat}
\bibliography{doriguello}

\end{document}